%% file: arxiv.tex
\documentclass[11pt]{article}

\usepackage[margin=1in]{geometry}
\usepackage[cmex10]{amsmath}

 \usepackage{graphicx}
 \usepackage{tikz}
 \usetikzlibrary{arrows,shapes}
\usepackage{amsthm}

\usepackage{amssymb}
\usepackage{amsmath}
\usepackage{multicol}
\usepackage{algorithm}
\usepackage{algorithmic}

\usepackage{url}
\usepackage{verbatim}
\usepackage{wrapfig}

\newcommand{\mypar}[1]{\medspace\noindent\textbf{#1}:}

\newtheorem{intf}{Interference Assumption}

\newtheorem{observation}{Observation}

\newtheorem{lemma}{Lemma}
\newtheorem{corollary}{Corollary}
\newtheorem{theorem}{Theorem}

\newcommand{\reals}{\mathbb{R}}

\def\calC{{\mathcal{C}}}

\def\calF{{\mathcal{F}}}

\def\calL{{\mathcal{L}}}
\def\calL{{\mathcal{L}}}

\def\calP{{\mathcal{P}}}

\newcommand{\prob}[1]{\textsc{#1}}

\newcommand{\ourprob}{\prob{Connectivity Scheduling}}
\newcommand{\steinerprob}{\prob{Steiner Connectivity Scheduling}}

\newcommand{\alg}[1]{\textsf{#1}}
\newcommand{\capalg}{\alg{CapKruskal}}

\newcommand{\defn}[1]{\underline{#1}}

\begin{document}

\title{Spanning Trees With Edge Conflicts and Wireless Connectivity\footnote{The first and last authors are
supported by grants nos.~152679-05 and 174484-05 from the Icelandic Research Fund. The second author is supported by NSF grant 1540547.}}

\author{
  Magn\'us M. Halld\'orsson \\
  ICE-TCS,  Reykjavik University \\
  \url{mmh@ru.is}
   \and
	Guy Kortsarz \\
  Rutgers University, Camden, NJ \\
  \url{guyk@camden.rutgers.edu}
   \and
	Pradipta Mitra \\
  Google Research, New York \\
  \url{ppmitra@gmail.com}
   \and
  Tigran Tonoyan \\ 
  ICE-TCS,  Reykjavik University \\
  \url{ttonoyan@gmail.com} 
}

\begin{titlepage}

\maketitle              

\begin{abstract}
We introduce the problem of finding a spanning tree along with a partition of the tree edges into fewest number of feasible sets, where constraints on the edges define feasibility.
The motivation comes from wireless networking, where we seek to model the irregularities seen in actual wireless environments.
Not all node pairs may be able to communicate, even if geographically close ---
thus, the available pairs are modeled with a link graph $\calL=(V,E)$.
Also, signal attenuation need not follow a nice geometric formulas ---
hence, interference is modeled by a conflict (hyper)graph $\calC=(E,F)$ on the links.
The objective is to maximize the efficiency of the communication, or 
equivalently minimizing the length of a schedule of the tree edges in the form of a coloring.

We find that in spite of all this generality, the problem can be approximated linearly in terms of a versatile parameter, the inductive independence of the interference graph. Specifically, we give a simple algorithm that attains a $O(\rho \log n)$-approximation, where $n$ is the number of nodes and $\rho$ is the inductive independence, 
and show that near-linear dependence on $\rho$ is also necessary.
We also treat an extension to Steiner trees, modeling multicasting, and obtain a comparable result.

Our results suggest that several canonical assumptions of geometry, regularity and ``niceness'' in wireless settings 
can sometimes be relaxed without a significant hit in algorithm performance.
\end{abstract}

\thispagestyle{empty}

\end{titlepage}

\section{Introduction}

We introduce the problem of finding a spanning tree along with a partition of the tree edges into fewest number of feasible sets, which are independent sets in a given conflict (hyper)graph.
The motivation comes from wireless networking, where we seek a basic communication structure while
capturing the irregularities seen in actual wireless environments. 

A spanning tree is the minimal structure for connecting the given set of nodes into a mutually communicable network.
The \emph{cost} of a communication spanning tree is the \emph{time} required to \emph{schedule} all the tree edges -- the transmission \emph{links} -- while obeying the \emph{interference} caused by simultaneous transmissions.

The \emph{scheduling complexity} of the tree represents its throughput capacity: how much communication can be sustained in the long run. The task might be to \emph{aggregate} the data measured at the sensor nodes, or to \emph{broadcast} using one-to-one communication to all nodes of the network.

Algorithmic studies of wireless connectivity to date have generally involved strong ``niceness'' assumptions.
One core assumption is that points are located in the Euclidean plane
and \emph{all (close enough) pairs of nodes are available as links} for use in the spanning tree.
Interference modeling has become progressively more realistic, starting with 
range-based graph models to the fractional \emph{SINR} model of interference, 
but the common thread is that \emph{interference is a direct function of the geometry}.
While natural, these assumptions depend on a simplified view of the nature of wireless communication.

Wireless networking in the real world behaves quite different from these theoretical models \cite{ganesan2002,kotz2004experimental,zamalloa2007} and typically displays a high degree of irregularity.
This manifests in how the strength of signals (and the corresponding interference) often varies greatly within the same region, and is often poorly correlated with distance \cite{baccour2012radio}.
This behavior holds even in simple outdoor environments, but is magnified inside buildings. It is also evidenced by fluctuations, sensitivity to environmental changes (even levels of humidity), and hard-to-explain unreliability.

There has been increased emphasis for greater robustness in the design and analysis of wireless algorithms to address the observed irregularities.
In the world of communications engineering, the default is to introduce stochastic distributions, e.g., on signal strengths. The algorithms world prefers more adversarial effects, but that can easily lead to intractability.

The objective of this work is to embrace this irregularity in connectivity problems.
We replace the previous assumptions by the opposite premises:
\begin{quote}
  \emph{A link may not be usable even if it should be.}  

\hspace*{-5ex} and 

  \emph{Interference need not follow (or even relate to) the underlying geometry.}
\end{quote}

Technically, the former premise means that the set of usable or \emph{available links} is now given as a \defn{link graph} $\calL=(V,L)$.
We place no restrictions on the structure of this graph.
The second premise implies another graph, this time on top of the \emph{links}.
Namely, the \defn{conflict (hyper)graph} $\calC=(L,F)$ specifies whether a given pair of links in $L$ can coexist in the same color (of a spanning tree).
In the \prob{Connectivity Scheduling} problem, we seek a spanning tree $T$ of $\calL$ and a coloring of the links of $T$ 
minimizing the number of colors used.

These formulations naturally raise a number of questions: Can arbitrary sets of available/usable links 
actually handled effectively?
Does it change anything? Can we actually disconnect the conflicts/interference from the geometry? Since the ugly specter of intractability is bound to raise its head somewhere, what are minimal restrictions that keep these problems well-approximable?

\mypar{Our Results}
Given the generality of the \prob{ConnectivityScheduling} problem, it is unsurprising that it is very hard even to approximate.
We show that strong $n^{1-\epsilon}$-approximation hardness holds, even for the natural special case of 2-hop interference.
Instead, we aim to obtain approximations in terms of natural parameters of the instances.

We show that the problem is approximable within $O(\rho \log n)$-factor, where $\rho$ is the \emph{inductive independence} of the (fractional) conflict graph. This is particularly relevant since $\rho$ is known to be constant in both of the predominant interference models: the physical (or SINR) model, and the protocol model. This is attained by a simple greedy algorithm that can be viewed as a combination of Kruskal's MST algorithm and a link scheduling algorithm for the physical model.
In contrast, we find that the (perhaps more natural) approach of selecting and coloring an MST fails badly.

We also generalize the problem to multicast, or Steiner trees, and obtain a similar logarithmic approximation (on graphs), involving a closely related (but slightly different) parameter.

\mypar{Definitions}
In line with a modern view of wireless interference, we represent the interference conflicts by 
\defn{fractional conflict graph} $\calC=(L,W)$.
Here $L$ is the set of communication links and $W: L \times L \rightarrow \reals^+$ 
is a function on ordered pairs of links, where $W(e,f)$ represents 
(or approximates) the degree to which a transmission on link $e$ interferes with a transmission on link $f$.
For convenience, let $W(e,e) = 0$.
We shall write $W(S,e) = \sum_{f\in S} W(f,e)$ and $W(f,S) = \sum_{e \in S} W(f,e)$.
Let $\calC[Y]=(Y,W\mathord{\upharpoonright}_Y)$ denotes the subgraph induced by a given subset $Y \subset L$.

A set $S$ of links is an \defn{independent} or a \defn{feasible set} if $W(S,e) \le 1$, for all $e \in S$. 
A \defn{coloring} of $\calC=(L,W)$ is a partition of $L$ into independent sets.
Observe that when $W$ is a 0-1 function, we have the usual independent sets and colorings of graphs.
Also, the fractional conflicts corresponds to certain hypergraphs that contain a hyperedge for each
minimal set $S'$ where $W(S',e) \ge 1$ for some $e \in S'$.

We can now state our {\ourprob} problem formally:
\begin{quote}
Given a link graph $\calL=(V,L)$ and a fractional conflict graph $\calC=(L,W)$,
we seek a spanning tree $T$ of $\calL$ and a coloring of $\calC[T]$, using the fewest number of colors.
\end{quote}

A fractional conflict graph $\calC=(L,W)$ is said to be $\rho$-\defn{inductive independent}, 
w.r.t.\ an ordering $\prec$ of the links, if
for every link $e$ and every feasible set $I\in \calF$ with $e \prec I$,  
$W(I,e)+W(e,I)\le \rho$, where $e \prec I$ means that $e$ precedes each link in $I$.
Here, ``inductive'' refers to how the interference is measured only towards later links, 
and ``independence'' that it is towards independent sets.
In geometric settings (including range-based and SINR models), $\prec$ corresponds to a non-decreasing ordering by link length.

For a fractional conflict graph $\calC=(L,W)$, let $\chi(\calC)$ denote the smallest number of independent sets into which $L$ can be partitioned;
when $\calC$ is an ordinary graph, $\chi(\calC)$ is the chromatic number of $\calC$.
We view a coloring of $\calC$ also as a \defn{schedule} and refer to the colors also as \defn{slots} (which could be time slots or frequency bands).

\mypar{Notable Instantiations}
{\ourprob} has a number of special cases of independent interest, both graph-based and geometric:

\begin{itemize}
\item When $\calC = L^2(\calL)$, two links conflict if they are incident on a common link.
This case corresponds to bidirectional version of the classic \emph{radio network} model.
The directed version of {\ourprob} was treated in \cite{GHKKO15} as the \emph{radio aggregation scheduling problem}.

\item In \emph{range-based or disk models}, nodes are embedded in the plane
and two links are adjacent if the distance between (the closest points on) them is less than $K$ times the length of the longer link, where $K$ is some fixed constant.
A variation measures the lengths from a particular node on each link. Also, in the the related \emph{protocol} model, adjacency occurs if the distance is less than $K_1$ times the length of the longer link plus $K_2$ times the length of the shorter link, for some constants $K_1, K_2$.

\item The original driving motivation is when nodes and links are 
embedded in a metric space and the fractional conflicts follow the \emph{geometric SINR model} of interference in terms of the lengths and distances between links.
Before this work, only the case when $\calL$ is the complete graph over a set of points in a Euclidean metric was considered.

\item A different geometric version is when we view that no interference at all is transmitted along unavailable links. The links are then unavailable because no signal gets transmitted between this pair of nodes, perhaps due to destructive alignment. We refer to this as the \defn{Missing Links} version.

\item A natural special case occurs when link unreliability is restricted by link length, so that only reasonably long links are unavailable or attenuated, but short links follow the normal SINR laws (short links are reliable). 

\item Finally, when the conflict graph $\calC$ is the \emph{line graph} of the link graph $\calL$,
i.e., $\calC = L(\calL)$, we obtain the well-known \emph{minimum degree spanning tree} (MDST) problem,
where given a graph $\calL$, the goal is to find a spanning tree of smallest maximum degree. 
By K\"onig's theorem, the chromatic number of the line graph of a tree (in fact, of any bipartite graph)
is equal to the maximum degree of the tree. 
This problem has more structure that allows for better solution: while it is NP-hard, it can be approximated within an additive one \cite{FurerRaghavachari}.
In particular, $L(\calL)$ is claw-free (which is stronger than being 2-inductive independence), and is intimately related to $\calL$.

\end{itemize}

\mypar{Related Work}
The connectivity problem in the geometric SINR model was first considered by Moscibroda and Wattenhofer \cite{MoWa06}.
In was, in fact, the first work on worst-case analysis in the SINR model.
They show that unlike in random networks, the worst-case connectivity depends crucially on the use of power control,
and with optimal power control, $O(\log^4 n)$ slots suffice to connect the nodes.
They soon improved this $O(\log^2 n)$ \cite{moscibroda06b,Moscibroda07}.
Currently, the best upper bounds known are $O(\log n)$~\cite{SODA12}
and $O(\log^* \Lambda)$ \cite{us:mobihoc17poster}, where $\Lambda$ is the ratio between the longest to the shortest length of a link in a minimum spanning tree (MST), 
a structural parameter that is independent of $n$.
Both of these results hold for the MST of the pointset; there are pointsets where $\Omega(\log^* \Lambda)$ slots are necessary for scheduling an MST \cite{us:mobihoc17poster}. 

The scheduling complexity of connectivity relates closely to the efficiency of \emph{aggregation}, a key primitive for wireless sensor networks. 
We refer the reader to~\cite{IncelGK11} for bibliography on  aggregation/collection problems.

There are many approaches that have been proposed to model irregularity in wireless networks.
We first examine static cases, or the modeling of non-geometric behavior.
The basic SINR model allows the pathloss constant $\alpha$ to be adjusted \cite{kumar00}, giving a first-order approximation of the signal gain.
In the engineering community, it is most common to assume that the deviations are drawn from a particular stochastic distribution,
typically assuming independence of events. 
In the theory camp, the prevailing approach is to view the variations as conforming the plane into a non-Euclidean metric space \cite{FKRV09,SODA11}, while retaining some tractable characteristics. This can also entail identifying appropriate parameters \cite{us:PODC14}. 

For frequent temporal changes, the standard engineering assumption is Rayleigh fading.
Dams et al.\ \cite{dams2015} (see also \cite{us:mobihoc17})
showed that link scheduling algorithms are not significantly affected by such variation, assuming independence across time.

For unpredictably changing behavior, there is much research on adapting to new conditions, particularly with exponential backoff. 
A theoretic model proposed to specifically capture unreliability is the 
\emph{dual graph model} \cite{kuhn2010broadcasting}, which extends the radio network model to a pair of graphs, the reliable and the unreliable links, where the latter are under adversarial control. The focus there is on distributed algorithms for one-shot problems, like global and local broadcast problems, where the nodes do not know which links are reliable. As far as we know, it has not been considered in settings 
involving a long-term communication structure.

Inductive independence was first defined by \cite{AADK02} and studied by \cite{YeB12} in the graph setting, while the weighted version
was introduced by Hoefer and Kesselheim \cite{HoeferKV14}. It has been used as a performance measure for 
various problems related to wireless networks, including
admission control \cite{GHKSV14}, dynamic packet scheduling \cite{K12packet,sicomp17},
and spectrum auctions \cite{HoeferKV14,HoeferK15,sicomp17}.

\mypar{Outline of the paper}
We first examine, in Sec.~\ref{sec:mst}, how the standard approach -- finding a minimum spanning tree -- fares for our problem,
and show that it can give poor solutions in every known interference model when there are missing or unreliable links.
We then give in Sec.~\ref{sec:greedy} a greedy algorithm for \prob{Connectivity Scheduling} achieving $O(\rho \log n)$-approximation, where $\rho$ is the inductive independence number of the conflict graph.
This dependence on $\rho$ is shown to be essentially tight in Sec.~\ref{sec:hardness}.
We also obtain a similar approximation of a \emph{Steiner or multicast} version of the problem in Sec.~\ref{sec:steiner}.

Implications of our results to the SINR (or physical) model are given in Sec.~\ref{sec:sinr}.
The rest of the paper can safely be read without any background in that model.
We then close with open problems.

We give a supplementary result in Sec.~\ref{sec:short}, involving natural geometric interference assumptions. A brief primer on SINR concepts is given in Appendix \ref{sec:sinrdefs}, for completeness.

\section{MST Fails}
\label{sec:mst}

In a basic setting, the nodes are located on the plane, and the interference between two links is a function of 
the lengths of links (distance between the two end-nodes), and the distance between the links. For instance, 
in the SINR model, the interference between two links is a decreasing function of their distance,
and an increasing function of the length of the interfered link. In this setting, the Euclidean minimum spanning
tree (MST) over the set of nodes is a natural candidate for connectivity, since it favors short links and has
low degree (or, more generally, contains few links in the vicinity of any node). Indeed, the 
MST of $n$ nodes can be scheduled in $O(\log n)$ slots in the Euclidean SINR model \cite{SODA12}.

Somewhat surprisingly, we find that when the set of possible links is restricted, the MST can actually fail quite badly. 
This holds in every reasonable model of interference.

\begin{intf} We say that an interference model is \defn{reasonable} if: a) incident links cannot be scheduled together,
while b) \defn{sparse} instances of equal length links can be scheduled in $O(1)$ slots, 
where a set of  length $\ell$ links is sparse if any ball of radius $\ell$ contains
$O(1)$ endpoints of those links.
\label{intf:reasonable}
\end{intf}
Every geometrically-defined wireless interference model known satisfies this reasonableness property.
In particular, this holds in the protocol and Euclidean SINR models.

\begin{theorem}
For any $n$, there is an instance of $n$ nodes embedded on the plane together with a spanning tree that is schedulable in $O(1)$ slots
while scheduling the MST requires $n^{1/3}$ slots, in every reasonable interference model.
\label{thm:mst}
\end{theorem}

\begin{wrapfigure}{r}{0.5\textwidth}
\scalebox{0.7}{\input{MSTfig.tikz}}
\end{wrapfigure}

\begin{proof}
Let $k\ge 1$ be a number and $K = 2 k^2$.
Let $V = \{o \} \cup \{v_{i,j} : i=0, 1, \ldots, k-1, j=0,1, \ldots, K -1\}$.
We position the nodes in the plane using polar coordinates, with the
node $o$ as the origin.
For node $v_{i,j}$, angular coordinate 
$r_{i,j}$ is $2\pi \cdot i/k$, while its radial coordinate is $k + j$.

The links are given by $L = O \cup T \cup Y$, where
$O = \{(o,v_{i,1}) : i=0,\ldots, k-1\}$,
 $T = \{(v_{i,j},v_{i,j+1}) : i=0,\ldots, k-1, j=0, \ldots, K-2\}$,
 $Y = \{(v_{i,K-1},v_{i+1 \bmod k ,K-1}) : i=0,\ldots, k-1\}$, 
or the \emph{ordinary}, the \emph{tiny} and the \emph{yuge} links.
That is, the link graph is in the form of a wheel, centered at origin, with $k$ spokes, and $K$ nodes on each spoke. Ordinary links are incident with the origin, while the yuge links form the tire of the wheel.

We observe that $d(v_{i,K-1},v_{i+1 \bmod k, K-1}) > k = d(o,v_{i',1})$, for any $i,i'$. Thus, the MST consists of the ordinary and tiny links, $S \cup L$. Since all the ordinary links have an endpoint in the origin, they must all be scheduled in different slots, implying that the MST requires $k = \Theta(n^{1/3})$ slots.

On the other hand, a more efficient solution is to use the set $Q$, consisting of $T$, $Y$ and one (arbitrary) link from $O$. 
This set $Q$ is sparse, and therefore can be scheduled in $O(1)$ slots.
\end{proof}

This same example shows why the known results for Euclidean SINR do not carry over to general metric spaces (even without missing links). Namely, one could simply form a metric space on the $n$ nodes by shortest-path distances in the link graph.

One way to try to overcome the hard example above would be to consider \emph{bounded degree minimum spanning trees}. However,
the example can be modified so that the maximum degree of the resulting link graph $\calL$ is at most $3$, but the result is
similar. To this end, one can replace the top vertex $o$ in the construction with a chain of $k$ equally spaced nodes connected into a simple path, where 
each node is incident with one ordinary link. The mutual distances between ordinary links are
still very small
compared with their lengths, and hence they must all be scheduled separately. 

This limitation result holds in interference models where very close links cannot be scheduled simultaneously.
This is a relaxation of Inteference Assumptions (IA) \ref{intf:sparsedense} introduced later.

\section{Greedy Algorithm}
\label{sec:greedy}

A natural greedy approach is to find a large feasible subset of edges, assign it a fresh color, and iterate on a contracted graph.
The key step is obtaining a constant-approximation for a maximum feasible subset. 
A logarithmic approximation then follows from a set cover argument.

We assume in this section that  $\calL$ can have parallel edges but no loops. We assume that the conflict graph 
$\calC$ is $\rho$-inductive independent for a number $\rho>0$, and the corresponding 
conflict function $W$ and ordering of edges $\prec$ are given. In the \defn{maximum feasible forest problem}, 
the goal is to find a maximum cardinality subset of edges of $\calL$, which is both independent in $\calC$ and acyclic in $\calL$.

 The algorithm, given as Alg.~\ref{alg:KK}, is a greedy Kruskal-like algorithm that mixes the edge selection criteria of wireless capacity 
algorithms \cite{SODA11,KesselheimSODA11} with the classic MST algorithm of Kruskal, thus the name {\capalg}. 
 It processes the edges in order of precedence $\prec$ and adds an
edge to the forest if: a) the interference on that edge from previously selected edges is small, and b) the edge does
not induce a cycle (as per Kruskal).  
We state it in terms of the classic union-find operations of \textsc{MakeSet},
\textsc{Connected}, and \textsc{Union}.

\begin{figure*}[htt!]
\hfill
 \begin{minipage}[t]{2.6in}
 \begin{algorithm}[H]
\caption{{\capalg}($\calL,\calC$)}          
\label{alg:KK}
\begin{algorithmic}[1]                    
  \STATE \textsc{MakeSet}$(v)$, for each $v \in V(\calL)$
   \STATE $S \leftarrow \emptyset$
   \FOR{$e = (u,v)$ in $L$ in $\prec$ order }
      \IF{$W(S,e)+W(e,S) \le 1/2$ and \NOT \textsc{Connected}$(u,v)$}
         \STATE $S \leftarrow S \cup \{e\}$
         \STATE \textsc{Union}$(u,v)$ \\
      \ENDIF
   \ENDFOR
   \RETURN $S' = \{e \in S : W(S,e) \le 1\}$
\end{algorithmic}
\end{algorithm}
 \end{minipage}
 \hfill
 \begin{minipage}[t]{2.2in}
\begin{algorithm}[H]
\caption{\textsc{Conn}($\calL, \calC$)}          
\label{alg:Conn}                           
\begin{algorithmic}[1]                    
   \STATE $i\leftarrow 0$
   \STATE $\calL_0\leftarrow \calL$
   \WHILE{$\calL_i$ has an edge}
      \STATE $S_i \leftarrow \capalg(\calL_i, \calC[\calL_i])$
      \STATE $\calL_{i+1} \leftarrow Contract(\calL_i,S_i)$
      \STATE $i \leftarrow i + 1$
   \ENDWHILE
   \RETURN $S_0, S_1, \ldots, S_{i-1}$
\end{algorithmic}
\end{algorithm}
 \end{minipage}
 \hfill
\end{figure*}

Recall that a subset $S$ of edges in $\calL$ is feasible if $W(S,e) = \sum_{f\in S} W(f,e) \le 1$, for all $e \in S$.
Define the ordered weight function $W^+$ as $W^+(e,f) = W(e,f)$ if $e \prec f$, and $W^+(e,f)=0$, otherwise.
Similarly, define $W^-$ as $W^-(e,f) = W(e,f)$ if $f \prec e$, and $W^-(e,f)=0$, otherwise.
Also define the cumulative versions $W^+(S,e)$, $W^+(e,S)$ as before.

We say that a set $S$ is \defn{semi-feasible} if for each $e\in S$, $W^+(S,e)+W^-(e,S)\le 1/2$.
Namely, the weighted indegree from shorter nodes and to longer nodes is bounded, but the total indegree of $e$ may not be.
By an averaging argument, a semi-feasible set $I$ contains a feasible subset of at least half its size.
Indeed, using semi-feasibility and sum rearrangements, we have,
\begin{equation}
\sum_{e\in S}W(S,e) = \sum_{e\in S}\left(W^+(S,e)+W^-(e,S)\right)  \le \frac{|S|}{2}
\label{eq:star}
\end{equation}
so for at least half of the links $e\in S$ it holds that $W(S,e)\le 1$.

\begin{theorem}\label{T:forest}
Let $F$ be a maximum feasible forest of $\calL$. Then ${\capalg}(\calL,\calC)$ outputs a feasible forest of size  $\Omega(|F|/\rho)$. 
\label{thm:kk}
\end{theorem}
\begin{proof}
Let $S$ and $S'$ be the sets computed in ${\capalg}(\calL,\calC)$.
By definition, $S'$ is feasible. 
To argue that $S'$ is large, we examine an arbitrary feasible forest, break it into three parts,
 and show that none of the parts can be too large compared to $S'$. 
 This will hold, in particular, for the optimal feasible forest.
By (\ref{eq:star}), we can focus on bounding $|S|$, as $|S'|\ge |S|/2$. 

Let $I$ be a feasible forest. Observe that the selection condition of the algorithm is equivalent to $W^+(S,e)+W^-(e,S) \le 1/2$, 
since the edges are considered in the order of $\prec$.
Let $I_R$ be those edges $e$ in $I$ that failed the degree condition  ($W^+(S,e)+W^-(e,S) > 1/2$),
and $I_T$ those edges $e=(u,v)$ in $I$ that failed the connectivity condition ($\textsc{Connected}(u,v)$).
The rest, $I_S = I \setminus (I_R \cup I_T)$ are contained in $S$.
We bound these sets in terms of $S$.

Since $I_T$ contains only edges inside components that $S$ also connects (recalling that $I$ induces a forest), $|I_T| \le |S|$.
Also, clearly $I_S \subseteq I \cap S \subseteq S$, so $|I_S| \le |S|$.
To bound the size of $I_R$, observe first that by the definition of $\rho$-inductive independence,
$W^-(I_R,f) + W^+(f,I_R) \le \rho$, for every edge $f\in S$. This implies that 
\[
  W^-(I_R,S)+W^+(S,I_R) = \sum_{f \in S} \left[W^-(I_R,f) + W^+(f,I_R)\right] \le \rho \cdot |S|.
\]
On the other hand, by the selection criteria, 
\[
  W^+(S,I_R)+W^-(I_R,S) = \sum_{e \in I_R} \left[ W^+(S,e)+W^-(e,S)\right]  
    > \sum_{e \in I_R} \frac{1}{2} = \frac{|I_R|}{2}.
\]
Thus, $|I_R| \le  2\rho \cdot |S|$ and $|I| \le (2\rho+2)|S|\le 4(\rho+1)|S'|$.
\end{proof}

\mypar{Connectivity Scheduling Algorithm}
The algorithm \textsc{Conn} repeatedly calls $\capalg$ to obtain a large independent set of links and 
assigns it to a new color class.
These links are then contracted and the process repeated until we have obtained a spanning tree.
 
The contraction of an edge is defined in the standard way, except we discard loops. Note that \emph{contraction leaves the conflict graph $\calC$ intact}.
The operation $Contract(\calL,S)$ contracts all edges in $S$ of a link graph $\calL$ and outputs the resulting graph.

The pseudocode of the algorithm is given in Alg.~\ref{alg:Conn}. The proof of the following theorem follows the classic set cover argument \cite{Johnson74a}.

\begin{theorem}\label{T:inductiveness}
  \textsc{Conn} terminates in $O(\rho \log n) \cdot \chi$ rounds, where $\chi$ is the number of colors needed for coloring an optimum spanning tree.
\label{thm:conn}
\end{theorem}

\begin{proof} Let $S_0,S_2,\ldots,S_i$ be the collection of edge-sets returned by $Conn$. For each index $k$, denote $s_k=|S_k|$, $n_k=|V(\calL_k)|$ and $x_k$ the cardinality of the optimum independent (in $\calC[\calL_k]$) forest in $\calL_k$.  Note that $\calC[\calL_k]$ is also $\rho$-inductive independent. Let $c\rho$ be an upper bound on the approximation ratio of {\capalg}, where $c>0$ is a constant. Hence, by Thm.~\ref{thm:kk}, 
\begin{equation}\label{E:kk}
n_k\ge \frac{x_k}{c\rho}.
\end{equation}
 Observe that $x_k\ge n_k/\chi$ (by the pigeonhole principle), and $n_k=n-\sum_{j<k}n_j$, since each iteration $j$ decreases the number of vertices by $|S_j|$ (as $S_j$ is a forest). Moreover, we can assume that $n_1,n_2,\ldots$ is a non-decreasing sequence, as otherwise we could rearrange the sets $S_k$ without violating (\ref{E:kk}). Thus, using monotonicity of $n_k$ and (\ref{E:kk}), we have
 \[
\frac{\sum_{j<k}n_j}{k-1}\ge n_k\ge \frac{n-\sum_{j<k}n_j}{c\rho \chi},
 \]
 so taking $k=\lceil c\rho \chi\rceil+1$, we see that $\sum_{j<k}n_j \ge n/2$. Namely, after every $\lceil c\rho \chi\rceil$ iterations the number of nodes is halved. This implies the required bound.
\end{proof}

\section{Reliable Short Links}
\label{sec:short}

We consider here the case when all short links are reliable. This is motivated by experimental results which indicate on one hand that signal strength is
poorly correlated with distance, but also that short links are nevertheless almost always strong and reliable
\cite{zamalloa2007}, with most of the variability in the links of intermediate range.
This is probably the most natural relaxation of the problem involving geometry.

The setting is as follows.
The nodes are located in the Euclidean plane.
There is a threshold distance, normalized to the unit distance, below which all links are reliable. Let $\Pi$ then denote the maximum link length that might be used, e.g., corresponding to the maximum distance at which signals can be properly received.
Then, node pairs of distance in the range 1 to $\Pi$ may or may not be available (= in the link graph),
while pairs within unit distance are all available. We call the links of length at most $1$ \emph{short} links.

We use limited assumptions about the interference model. 
We first define some notions.
By a \defn{$t$-square} we mean a square of side $t$ in the plane.
A square \defn{hits} an edge (or link) if an endpoint of the edge is within the square.
A set of links of length \emph{at most} $\ell$ is said to be \defn{$s$-sparse} if every $\ell$-square 
hits at most $s$ links, and
a set of links of length \emph{at least} $\ell$ is \defn{$d$-dense} if some  $\ell$-square 
hits at least $d$ links.

\begin{intf}
A $s$-sparse set of links can be scheduled in $O(s)$ slots, while a $d$-dense set requires $\Omega(d)$ slots.
\label{intf:sparsedense}
\end{intf}
These assumptions are satisfied by all major interference models defined in the plane (or in doubling metrics); we will argue this for the SINR model in Sec.~\ref{sec:sinr}.

We examine how the approximability of the problem varies with $\Pi$.
It turns out that minimum spanning trees (MST) work well here.

\begin{theorem}
The non-short links of an MST can be scheduled in $O(a)$ slots, 
where $a = \Pi\sqrt{\chi}$ and $\chi$ is the optimum number of slots of a spanning tree.
\label{thm:shortgood}
\end{theorem}

In the following discussion, we will work with a fixed MST $T$ of the link graph.

We first note that an upper bound of $O(\Pi^2)$ slots holds for the length of any non-improvable schedule for $T$.
We refer to the maximal connected subgraphs of $T$ containing only short links as \defn{clusters}.
A $t$-square \emph{hits} a cluster (or link) if it contains a vertex of that cluster (link).
A $1/2$-square can hit at most one cluster, as otherwise the respective endpoints would be within unit distance and could be connected by a short link. Thus, a given $\Pi$-square $S$ hits at most $(2\Pi+1)^2$ clusters, since it can be covered with that many $1/2$-squares. Note that the non-short links of $T$ that are hit by $S$ are all used to connect clusters hit by the $3\Pi$-square that has $S$ in the center. 
Hence, at most $9\cdot (2\Pi+1)^2$ non-short links are hit by $S$, since there are at most that many clusters hit by $S$, and $T$ is a tree. Thus, non-short links are $O(\Pi^2)$-sparse, and can be scheduled in $O(\Pi^2)$-slots, by IA \ref{intf:sparsedense}.

To obtain a tighter bound, we split the non-short links into \emph{medium} links, of length from $1$ to $\sqrt{a}$,
and \emph{long}, that are longer than $\sqrt{a}$, and treat these separately. The case of medium links is easy. Applying the argument from the previous paragraph to a $\sqrt{a}$-square and medium links hit by it, we see that the medium links are $O(a)$-sparse, and can be scheduled in $O(a)$ slots, by IA~\ref{intf:sparsedense}. The case of long links needs a more delicate argument. We refer to the maximal connected subgraphs of $T$ containing only \emph{non-long} links as \defn{blocks}.
The difficult part is to show that every $\Pi$-square hits $O(a)$ blocks. This is done in Lemma~\ref{lem:longlinks}. The rest of the proof of Thm.~\ref{thm:shortgood}, showing that long links of $T$ can be scheduled in $O(a)$ slots, follows along the same lines as above.

\begin{lemma}
Each $\Pi$-square $S$ hits $O(a)$ blocks.
  \label{lem:longlinks}
\end{lemma}

\begin{proof}
 We will account for the number of blocks hit by $S$, by reasoning about their relation to some fixed optimal spanning tree $T_{OPT}$. First, observe that every block hit by $S$ must have a vertex incident with a long edge (of length at least $\sqrt{a}$)  in $T_{OPT}$. Otherwise, such a block would be connected to the remaining vertices of the graph by a non-long edge in $T_{OPT}\setminus T$, which would contradict its maximality and the fact that $T$ is a MST.
  We classify the blocks hit by $S$ as Class 1, containing a vertex inside $S$ which is incident to a long link in $T_{OPT}$, and Class 2, the remaining blocks. We bound the two classes separately.

 Let $t$ denote the number of long edges in $T_{OPT}$, hit by $S$.
  Since $S$ can be covered with $O((\Pi/\sqrt{a})^2) = O(\Pi/\sqrt{\chi})$ $\sqrt{a}$-squares, one of them hits at least $\Omega(t/(\Pi/\sqrt{\chi}))= \Omega(t \sqrt{\chi}/\Pi)$ long edges. Thus, the long edges  of $T_{OPT}$ are $t \sqrt{\chi}/\Pi$-dense, and so by
 IA~\ref{intf:sparsedense}, $\chi = \Omega(t \sqrt{\chi}/\Pi)$. Rearranging, we have that $t = O(a)$. 
This trivially implies that the number of Class 1 blocks hit by $S$ is in $O(a)$.  
  
Next, we consider Class 2 blocks. 
Each such block must have a vertex that is incident to a long edge $e$ in $T_{OPT}$, with both its endpoints outside of $S$. 
Thus, a Class 2 must have vertices both inside and outside $S$, and $T_{OPT}$ uses only short or medium links to connect vertices from the two sides. 
For each Class 2 block $B$, identify a single link used in $T_{OPT}$ to connect the vertices of $B$ inside $S$ to those outside $S$, and refer to it as $B$'s \defn{linker}. Note that each $1/2$-square hits at most one linker, as otherwise the corresponding blocks could be connected with a short edge, contradicting maximality.
\emph{Short} linkers are of length at most $\sqrt{\chi}$,
and hence must have an endpoint in $S$ within distance $\sqrt{\chi}$ from the border of $S$, as they must cross the border. Thus, the total area in $S$ that can contain an endpoint of a short linker is at most $2\Pi \sqrt{\chi}$, and by covering it with $1/2$-squares, we see that there can be at most $O(\Pi \sqrt{\chi})$ short linkers.

We partition the non-short linkers into \defn{$i$-linkers}, of length between $Q:=2^i \cdot \sqrt{\chi} $ and $2Q=2^{i+1}\cdot \sqrt{\chi}$, for $i=0,1, \ldots$.
Let $q_i$ be the number of Class 2 blocks with $i$-linkers.
Observe that an $i$-linker has an endpoint in $S$ within distance $2Q$ from the border of $S$. Thus, the total area in $S$ that can contain an $i$-linker is less than $2\Pi Q$, and
can be covered with $O(\Pi/Q)$ different $Q$-squares.
Thus, some $Q$-square hits $\Omega(Q/\Pi)$ $i$-linkers (each in $T_{OPT}$, and of length at least $Q$),
so $T_{OPT}$ is $\Omega(q_i Q/\Pi)$-dense. Hence, by IA~\ref{intf:sparsedense}, $\chi = \Omega(q_i Q/\Pi)$,
and by rearranging, $q_i = O(\sqrt{\chi} \Pi/2^i) = O(a/2^i)$.
 The total number of Class 2 blocks is then bounded by
 $\displaystyle O(\Pi\sqrt{\chi}) + \sum_{i=0} q_i = O(a) + \sum_{i=0} O(a/2^i) = O(a) \sum_{i=0} 2^{-i} 
  = O(a)\ .$
\end{proof}

The short links are contained in an MST of the complete graph on the pointset. 
Thus, we can derive the following bound on the total number of slots used.

\begin{corollary}
An MST can be scheduled in $\zeta + O(\Pi\sqrt{\chi})$ slots, where
$\zeta$ is the number of slots required in the complete graph setting.
\end{corollary}

In many settings, $\zeta$ is a negligible term, in which case we obtain a $\Pi$-approximation. 
For Euclidean SINR, $\zeta = O(\min(\log n, \log^* \Lambda))$, where $\Lambda$ is the ratio between the length of the longest and the shortest possible link \cite{SODA12,us:mobihoc17poster}.

\mypar{Limitations.}
This bound on the MST is in fact best possible, by the result of Sec.~\ref{sec:mst}.
Namely, in the construction of Sec.~\ref{sec:mst}, 
the threshold under which all links are available is $1$, while $\Pi < 2k$.
As the MST used at least $k$ slots, it amounts to $\Omega(\Pi)$.
On the other hand, the instance is $O(1)$-schedulable, and there are no short links used.

\begin{observation}
There is an instance with $\chi = O(1)$ and no short links, for which an MST requires $\Omega(\Pi)$ slots.
\end{observation}


\mypar{More General Interference Models}
We can also incorporate \emph{missing links} by relaxing the interference assumptions to:
\begin{intf}
A $d$-dense set requires $\Omega(d/\rho')$ slots.
\label{intf:sparsedense2}
\end{intf}
This would, for instance, hold when the conflict graph is a subgraph of a disc graph and has (unweighted) inductive independence number $\rho'$.
Our results then hold with an additional $\rho'$ factor.

\begin{corollary}
An MST can be scheduled in $\zeta + O(\rho' \Pi \sqrt{\chi})$ slots, where
$\zeta$ is the number of slots required in the complete graph setting.
\end{corollary}

The results can also be easily transferred to other doubling metrics satisfying IA \ref{intf:sparsedense}.  The asymptotic $\Pi$-approximation factor then becomes $\Pi^{d/2}$, where $d$ is the doubling dimension of the metric (details omitted).

\section{Multicast Tree Schedules}
\label{sec:steiner}

A natural generalization of {\ourprob} is to allow for a set of optional nodes that can be used in the tree construction but need not. Formally, the node set $V$ contains a subset $X$ of terminals and we seek a Steiner tree that spans all the terminals. As before, we ask also for the shortest schedule of the tree links. We refer to this as the {\steinerprob}.

It is not hard to construct examples for which optimal multicast trees are arbitrarily better than trees that use only the terminals, even in a geometric setting. One instance can be obtained from the example of Sec.~\ref{sec:mst} by restricting the terminals to only the origin and the nodes incident on yuge links.

We give an algorithm for {\steinerprob} with \emph{unweighted} conflict graph $\calC$, and analyse is in terms of a parameter similar to $\rho$ but involving clique covers rather than independence. An unweighted graph $\calC=(L,E)$ is \defn{$\eta$-simplicial} if there is an ordering $\prec$ of $L$ such that for each link $v \in L$, the subgraph induced by $v$'s neighbors that are later in the ordering can be covered with $\eta$ cliques.
We refer to neighbors later in the ordering as \defn{post-neighbors}.
As before, in the geometric setting, the ordering is given by link length.
Observe that $\rho \le \eta$, while the best bound in the other direction is $\eta \le \rho \log n$.

Our algorithm is a reduction to a multi-dimensional version of the Steiner tree (MMST) problem, recently treated by Bil\`o et al.\ \cite{Bilo17}. In MMST, each edge of the input graph has an associated $d$-dimensional weight vector, where the weight of edge $e$ along dimension $i$ indicates how much of the $i$-th \emph{resource} is required by $e$.
The objective is to find a tree that minimizes the $\ell_p$-norm of its load vector, 
where the load vector of a Steiner tree is the sum of the weight vectors of its edges.
We use here the $\ell_{\infty}$-norm, as we want to minimize the maximum use of a resource. They give a greedy $O(\log d)$-approximation algorithm for that case.

Given an instance of {\steinerprob} with link graph $\calL$ and conflict graph $\calC$, our reduction is as follows. Each link $e$ in $\calL$ is  itself (or corresponds to) a resource, so there are $n$ (=number of edges) resources. The weight of link $f$ along dimension $e$ is 1 if $f$ is a post-neighbor of $e$ in the conflict graph $\calC$, and 0 otherwise.

Suppose now that the MMST algorithm of \cite{Bilo17} returns a tree $T$ with $\ell_\infty$-norm $Z$.
Then, the sum of the tree edges along each dimension is at most $Z$, namely, each link (whether in $T$ or not) has at most $Z$ neighbors in $T$. In particular, $\calC[T]$ is $Z$-inductive, and can then be colored greedily using $Z+1$ colors.

On the other hand, consider an optimal tree $T^*$ and let $Z^*$ denote the infinity norm of its load vector.
From \cite{Bilo17}, we know that $Z = O(\log n)\cdot Z^*$.
Let $f$ be a link with $Z^*$ post-neighbors in $T^*$, and let $N_f$ be its set of post-neighbors in $T^*$. By assumption, $\calC[N_f]$ can be covered with $\eta$ cliques,
and thus $N_f$ contains a clique of size at least $|N_f|/\eta = Z^*/\eta$.
It follows that the length of the schedule of the optimal tree is at least the
chromatic number of $\calC[N_f]$, which is at least $Z^*/\eta$.
Thus, our solution yields a $O(\eta \log n)$-approximation.

\begin{theorem}
There is a $O(\eta \log n)$-approximation algorithm for {\steinerprob}, where conflicts are given by a $\eta$-simplicial graph.
\end{theorem}

It is a folklore that $\eta\le 6$ in disk graphs.

\begin{corollary}
There is a $O(\log n)$-approximation algorithm for {\steinerprob}, where conflicts are given by a disk graph.
\end{corollary}

\section{Hardness of Approximation}
\label{sec:hardness}

It is easy to see that with an arbitrary conflict graph $\calC$, the problem is  hard to approximate.
For instance, if the link graph $\calL$ is already a spanning tree, {\ourprob} becomes simply the classical graph coloring problem (of $\calC$).
We show below that the hardness extends to other more restricted settings.
These results also show that near-linear dependence on $\rho$, the inductive independence, is unavoidable.

We first show that hardness holds when $\calC$ is the square of the line graph of $\calL$, $\calC = L^2(\calL)$.
This corresponds to (bidirectional) 2-hop interferences: two transmission links conflict if they are incident on a common edge.
The reduction is from the \textsc{Distance-2 Edge Coloring} problem in general graphs, 
also known as \textsc{Strong Edge Coloring}: Given a graph $\calL$,
find a partition of the edge set into induced matchings, i.e., induced subgraphs where every vertex is of degree 1.

\begin{theorem}
The {\ourprob} problem is hard to approximate within $n^{1-\epsilon}$-factor, for any $\epsilon > 0$, even when $\calC = L^2(\calL)$.
\label{thm:graphhard}
\end{theorem}

\begin{proof}
Given an instance of  \textsc{Strong Edge Coloring} with graph $G'=(V',E')$, we construct an instance of {\ourprob} problem with the graph $\calL$ constructed as follows. Consider a bipartite graph $G''=(V_1,V_2, E)$, as follows. For each vertex $v$ in $V'$, there are two vertices $v_1, v_2$ in $V=V_1\cup V_2$, where $v_i\in V_i$, $i=1,2$.
If $uv \in E'$ then $v_1u_2$ and $v_2 u_1$ are in $E$.
Link graph $\calL$ is obtained from $G''$ by taking a complete binary tree with $|V_2|$ leaves and identifying each leaf with a vertex of $V_2$. The conflict graph is given by a simple graph $\calC$ with vertex set $E$, where $e_1,e_2\in E$ are adjacent in $\calC$ if and only if they form an induced matching in $G'$, i.e., there is no edge in $G'$ connecting an endpoint of $e_1$ to an endpoint of $e_2$. This completes the construction.

First, let us show that a strong edge coloring of $G'$ can be used to construct a spanning tree in $\calL$ with a similar coloring number. Consider a strong coloring that partitions the edges of $G'$ into $c$ color classes $E_1, E_2, \ldots, E_c$. 
Each class $E_i$ induces a pair of feasible slots $S_i$, $S'_i$ in $\calL$, where $S_i = \{v_1u_2 : uv \in E_i\}$ and $S'_i = \{u_1v_2 : uv \in E_i\}$. 
Indeed, since $E_i$ is an induced matching in $G'$, each of the slots $S_i,S'_i$ is also an induced matching in $\calL$ (and hence independent in $\calC$). Note that the edges in these slots cover all vertices of $\calL$, except for the binary tree.
We also add $O(\log n)$ slots to the schedule, two for each layer in the binary tree. The number of slots used then is $O(c + \log n)$. This gives us a connected subgraph of $\calL$ that can be scheduled in $O(c+\log n)$ slots.

Next, consider a  spanning tree of $\calL$ with a corresponding schedule of the edges in slots $S_1, S_2, \ldots, S_t$. Ignoring all edges within the binary tree, we obtain a partition of the edges of the bipartite graph $G''$ between $V_1$ and $V_2$. We claim that each class corresponds to an induced matching in $G'$,  leading to a strong edge coloring of $G'$ with $t$ colors.

Consider a pair of edges $v_1u_2$ and $w_1x_2$ in the same feasible slot.
Since they are feasible, there are no edges $v_1x_2$ nor $w_1u_2$ in $\calL$, and thus no edges $xw$ nor $vu$ in $E$. 
Then $vu$ and $wx$ form an induced matching in $G'$.

Hence, the optimum number of colors in strong edge coloring of $G'$ is within a constant factor plus a logarithmic term of the optimal number of slots needed for scheduling a spanning tree in $\calL$. 
Since the former is hard to approximate within $n^{1-\epsilon}$-factor \cite{chalermsook2013}, so is the latter.
\end{proof}

\begin{theorem}
The graph variant is $n^{1-\epsilon}$-hard to approximate, for any $\epsilon > 0$.
This holds even if the link graph $\calL$ is complete.
\end{theorem}

\begin{proof}
We add to instance of Thm.~\ref{thm:graphhard} all edges that were not there and make them adjacent to \emph{all} other edges in the graph. If these new edges are used in a spanning tree, they have to be scheduled separately in individual time slots.
Thus, using them does not decrease the length of any schedule.
\end{proof}

\section{Implications to Signal Strength Models}
\label{sec:sinr}

We consider in this section the implementation and implication of our results to signal strength models, most importantly metric SINR model.

SINR-feasibility, besides the underlying metric, also depends on the \emph{transmission power control} regime. Different power control methods give different notions of feasibility. Nevertheless, it is known that for most interesting cases, SINR-feasibility has constant-inductive independence property. In particular, power control is usually split into two modes: \defn{fixed monotone power schemes}, where links use only local information, such as the link length, to define the power level, and \defn{global power control}, where all power levels are controlled simultaneously to give larger independent sets. The former includes the uniform power mode, where all links use equal power. Another technical issue is \emph{directionality} of links, which is not explicitly addressed by our general results, but will be addressed below.

Let us start the discussion from Euclidean metrics (or more generally doubling metrics). For the global power control mode, \cite{KesselheimSODA11} introduced a weight function $W$ and proved that with this function, the conflict graph of any set of links 
is constant-inductive independent (see \cite[Thm. 1]{KesselheimSODA11}), so our results apply here directly (except for directionality issues, addressed below). Similarly, for fixed monotone power schemes (excluding uniform power), \cite{sicomp17} showed that in order to get constant-inductive independence, one may take the natural weight function,  \defn{affectance} (also called \emph{relative} or \defn{normalized interference})~\cite[Thm.3.3]{sicomp17}. In all cases, the ordering $\prec$ corresponds to a non-decreasing order of links by length. 

For general metric spaces, a slightly more technical definition of inductive independence is used, where a fractional conflict graph $\calC=(L,W)$ is $(\rho,\gamma)$-inductive independent, 
w.r.t.\ an ordering $\prec$ of the links, if
for every link $e$ and every feasible set $I\in \calF$ with $e \prec I$, there is a subset $I'\subseteq I$ of size $|I'|\ge |I|/\gamma$, such that $W(I,e)+W(e,I)\le \rho$. The old definition corresponds to the setting $\gamma=1$. It is easily verified that Thms.~\ref{T:forest} and~\ref{T:inductiveness} extend to cover this new definition, with approximation ratios multiplied by a factor of $\gamma$. Now, the counterparts of the results from the previous paragraph in general metrics can be found in~\cite[Lemmas 2,4]{icalp11} and \cite[Thm.~1, Lemma 3]{KesselheimESA12}, where it is shown that with appropriate weight functions, feasibility for any fixed monotone power scheme (including uniform power), as well as feasibility with global power control, can be expressed by a fractional conflict graph, which is ($O(1),O(1)$)-inductive independent.

The claims above concern settings where the links have fixed directions. In particular, if we apply Thm.~\ref{T:inductiveness} to the weighted functions from the previous paragraph, then we should add ``there exists a direction of links, such that...'' to the claim. This issue is easily resolved for the global power control mode, where the weight function of \cite{KesselheimSODA11} does not depend on directions. Namely, it gives a schedule, such that whatever direction is assigned to the links, one can find a power assignment that makes it work (the power assignment could be different for different orientations of links).

For oblivious powers, the following trick applies. It is known that  for a set of links with some direction and an oblivious power assignment, and with the weight function $W$ defined in terms of the affectances, if $W(e,S)\le 1/2$ for all $e\in S$ (call this \defn{dual-feasibility}), then there is another oblivious power assignment (called the \emph{dual} of the original one) that makes $S$ feasible with the reversed directions of links~\cite{KV10}. 
Thus, we would like to have schedules with slots $S$ being also dual-feasible. To this end, it is enough to modify {\capalg}, so that the threshold $1/2$ in the acceptance condition is replaced with $1/4$, and the output set $S'$ is given by $S'=\{e\in S: (W(S,e)\le 1) \vee (W(e,S)\le 1/2)\}$. Very similar methods then show that this again gives an $O(\rho)$-approximation to the maximum feasible forest problem. The rest of the analysis is left intact, so we obtain an $O(\log n)$-approximation as before, but with schedule slots that are both feasible and dual-feasible. Then we can replace each slot with its two copies and revert the directions of links in one of the copies. Every link thus gets scheduled in both directions, while the schedule length increases by a factor of two.

Summarizing the observations above, we state the following theorem.

\begin{theorem}
  There is an $O(\log n)$-approximation to {\ourprob} problem in the SINR model in arbitrary metric spaces.
This holds both in the case of fixed monotone power assignments,
and for arbitrary power control.
It holds even when only a subset of the node-pairs are available as links (but interferences follow the metric SINR definitions).
\end{theorem}

These are the first results that hold in general metrics.
They are necessarily relative approximations, since in general metric spaces, there is no good upper bound on the
connectivity number, even for complete graphs.  Two simple examples
are the metric induced by the star $K_{1,t}$ with unit-length edges,
and the unit metric formed by distances on the unit-length clique metric.

\smallskip

For the case of points in the plane (i.e., a complete link graph with conflicts induced by distances), connectivity can be
achieved in $O(\log n)$ slots \cite{SODA12}. Since it is not known if $O(1)$ slots always suffice, this result is not
directly implied by Thm.~\ref{thm:conn}. However, it was also shown in \cite{SODA12} that the MST contains a feasible forest of $\Omega(n)$ edges. The rest of our analysis (using constant-inductive independence) then implies a result matching \cite{SODA12}.

\begin{corollary}
  Let $P$ be a set of points in the plane.  Then, \textsc{Conn} finds and schedules a spanning tree of $P$ in $O(\log n)$ slots.
\end{corollary}

\mypar{Steiner trees}
In the geometric SINR model with a fixed monotone power scheme (with not all links available), we reduce the problem to a graph question as follows.
It was observed in \cite{us:talg12} that links of the same length class behave approximately like unit-disk graphs, where a length class refers to links whose lengths differ by at most a factor of $2$.
Namely, 
there are constants $c_1$ and $c_2$ such that 
for a set $S$ of links of length approximately $\ell$, if all links are of mutual distance greater than $c_2 \ell$,
then they form a feasible set, whereas any pair of links in $S$ of distance at most $c_1 \ell$ must be scheduled separately.

We modify the reduction to MMST to that of the 
graph construction so that 
weight of link $f$ along dimension $e$ is 1 only if $f$ is a post-neighbor of $e$ in $\calC$ \emph{and} $f$ and $e$ are of the same length class.
We then take the resulting tree and schedule the length classes separately, at an extra cost of $O(\log \Lambda)$ (the number of length classes).

\begin{corollary}
There is a $O(\log \Lambda \log n)$-approximation algorithm for {\steinerprob} in the geometric SINR model,
under  any fixed monotone power scheme.
\end{corollary}

Using power control, we can do considerably better.
The main result of \cite{us:stoc15} shows that for any set $L$ of links, there is an unweighted conflict graph $\calC(L)$, such that every independent set in $\calC$ is feasible, and the chromatic number of $\calC$ is at most $O(\log^*\Delta)$ factor away from the optimum schedule length of $L$ (using global power control). Moreover, $\calC$ is constant-simplicial~\cite[Prop. 1]{us:stoc15}.

\begin{corollary}
  There is a $O(\log n \log^* \Lambda)$-approximation algorithm for {\steinerprob} in the geometric SINR model
with global power control.
\end{corollary}

A similar result with $O(\log\log \Lambda)$-factor holds also for certain monotone power schemes (but not, for instance, uniform power) \cite{us:fsttcs15}.

\mypar{Hardness}
A special Missing Links variant of the geometric case is where the nodes/links are embedded in the plane
and all interferences are either zero or follow the SINR model (with either fixed power or global power control).

\begin{theorem}
The geometric Missing Links variant is $n^{1-\epsilon}$-hard to approximate, for any $\epsilon > 0$.
It is also $\Lambda^{2-\epsilon}$-hard, where $\Lambda$ is the ratio between the longest to the shortest node distance.
This holds even if all unavailable links are missing links.
\label{thm:missingharddelta}
\end{theorem}

\begin{proof}
We embed the instance of the previous theorem in the plane.
The nodes of $V_1$ are located in a unit square in a mesh pattern, $1/\sqrt{n}$ apart in $\sqrt{n}$ columns $\sqrt{n}$ abreast. 
At a unit distance, a similar unit square holds the nodes of $V_2$.
The length of an edge in $\calL$ (in distance in the plane) is then between 1 and 4.

An induced matching in $\calL$ corresponds to a set of links with no mutual interference.
On the other hand, a pair of links that are incident on a common edge or share a vertex, will receive interference from each other according to the SINR formula (using the shared edge or each other). Given that distances along available edges vary only by a constant factor, the interference between the links is a constant (specifically, at least $1/4^\alpha$, where $\alpha$ is the ``pathloss'' constant of the SINR model).
Thus, in the setting where the SINR threshold is at least the reciprocal of that constant (i.e., $\beta \ge 4^\alpha$), feasible sets are necessarily induced matchings in $\calL$. We can then conclude by recalling a ``signal-strengthening'' result \cite{GHW14} that shows that varying the threshold by a constant factor only affects the schedule length by a constant factor.

The longest node distance is at most $\log n$, which is from the root of the binary tree to its leaves,
while the shortest distance is $1/\sqrt{n}$. Thus, 
$\Lambda \le 4 \log n \sqrt{n}$, and $n^{1-\epsilon} \ge \Lambda^{2-\epsilon'}$, for some $\epsilon' \ge \epsilon/3$.

We can restrict the available edges incident to (non-leaf) nodes on the binary tree to the tree edges alone. 
Thus, non-leaf nodes in the tree must be connected via the tree edges. Then, all unavailable edges are missing edges.
\end{proof}

\section{Open Issues}
\label{sec:open}

Many related problems are left addressing; we list the most prominent ones.
\begin{itemize}
\item Latency minimization: Bounding the time it takes for a packet to filter through the tree from a leaf to a root (and back).
This requires optimizing both the height of the tree as well as the ordering of the links in the schedule.
\item Directed case: Finding an arborescence. This requires new techniques, as our argument crucially depends on the graph being undirected.
\item Distributed algorithms. This relates also to the issue of detecting or learning whether a link is usable/reliable or not.
\end{itemize}

\bibliographystyle{abbrv}
\bibliography{references}

\appendix

\section{SINR Definitions}
\label{sec:sinrdefs}

For completeness, we include here various definitions and facts regarding the SINR model.

The \defn{abstract SINR} model has two key properties: \textbf{(i)} signal decays as it travels from a sender to a
receiver, and \textbf{(ii)} interference -- signals from other than the intended transmitter -- accumulates.
Transmission succeeds if and only if the interference is below a given threshold.
The \defn{Metric SINR} model additionally assumes
\emph{geometric path-loss}: that signal decays proportional to a fixed polynomial of the distance, where the \emph{pathloss constant} $\alpha$ is assumed to be an arbitrary
but fixed constant between 1 and 6.  This assumption is valid with $\alpha=2$ in free space and perfect vacuum
\cite[Sec.~3.1]{Goldsmith}.
In the \emph{Euclidean SINR} model, the distances are planar.

Formally, a \emph{link} $l_v = (s_v, r_v)$ is given by a pair of nodes, sender $s_v$ and a receiver $r_v$, which are located in a metric space. Let $d(x,y)$ denote the distance between points $x$ and $y$ in the metric, and use the shorthand $d_{vw} = d(s_v,r_w)$.
The strength of a signal transmitted from point $x$ as received at point $y$ is $d(x,y)^\alpha$.
The \emph{interference} $I_{uv}$ of sender $s_{u}$ (of link $l_u$) on the receiver $r_v$ (of link $l_v$) is $P_u / d_{uv}^\alpha$,
where $P_v$ is the power used by $s_v$.  When $u=v$, we refer to $I_{vv}$ as the \emph{signal strength} of link $l_v$.
If a set $S$ of links transmits simultaneously, then the \emph{signal to noise and interference ratio} (SINR) at $l_v$
is
\begin{equation}
 \text{SINR}_v := \frac{I_{vv}}{N + \sum_{u \in S} I_{uv}} = 
   \frac{P_v / d_{vv}^\alpha}{N + \sum_{u \in S} P_v / d_{uv}^\alpha}\ ,
\label{eqn:sinr}
\end{equation}
where $N$ is the ambient noise.
The transmission of $l_v$ is \emph{successful}
iff $\text{SINR}_v \ge \beta$, where $\beta \ge 1$ is a hardware-dependent constant.  

\paragraph*{Additional definitions: Power, affectance, separability}

We will work with a total order $\prec$ on the links, where $l_v \prec l_w$ implies that $d_{vv} \le d_{ww}$. 
A power assignment $\calP$ is \emph{monotone} if both $P_v \le P_w$
and $\frac{P_w}{d_{ww}^\alpha} \le \frac{P_v}{d_{vv}^\alpha}$ hold whenever $l_v \prec l_w$.
This captures the main power strategies, including uniform and linear power.

The affectance $a^{\calP}_w(v)$ \cite{GHW14,KV10} of link $l_w$ on link $l_v$ under power
assignment $\calP$ is the interference of $l_w$ on $l_v$ normalized to
the signal strength (power received) of $l_v$, or
\[ a_w(v) = \min \left(1, c_v \frac{P_w}{P_v} \frac{d_{vv}^\alpha}{d_{wv}^\alpha}\right)\ ,\]
where $c_v= \frac{\beta}{1-\beta N/(P_v/d_{vv}^\alpha)} > \beta$ is a factor depending only on universal constants and the
signal strength $P/d_{vv}^\alpha$ of $l_v$, indicating the extent to which the ambient noise affects the transmission.  We drop
$\calP$ when clear from context.  Furthermore let $a_v(v) = 0$. For a set $S$ of links and link $l_v$, let $a_v(S) =
\sum_{l_w \in S} a_v(w)$ be the \emph{out-affectance} of $v$ on $S$ and $a_S(v) = \sum_{l_w \in S} a_w(v)$ be the
\emph{in-affectance}.
Assuming $S$ contains at least two links we can rewrite Eqn.~\ref{eqn:sinr} as $a_S(v) \leq 1$ and this is the form we
will use.  A set $S$ of links is \emph{feasible} if $a_S(v) \leq 1$ and more generally \emph{$K$-feasible} if $a_v(S)
\le 1/K$.

The following theorem shows that the interference model assumptions of Sections \ref{sec:mst} and \ref{sec:short} hold for geometric SINR. This fact is widely known, see e.g., \cite{SODA12}. We outline a proof for completeness.
\begin{theorem}[\cite{SODA12}]
If a link set is $s$-sparse, then it can be scheduled in $O(s)$ slots in geometric SINR,
and if it is $d$-dense, then it requires $\Omega(d)$ slots.
\label{thm:sparsedense}
\end{theorem}

\begin{proof}
The former claim essentially follows from the results of~\cite{us:talg12}. Here is a crude sketch of a proof. Let $L$ be a $s$-sparse set of links of length at most $\ell$. Partition the plane into squares of side $\ell$. Assign each link to a square where it has an endpoint, ties broken arbitrarily. It is easy to color the squares using constant number of colors, such that for each color class $\calC$, the distances between the squares in $\calC$ are greater than $c\ell$, where $c$ is a constant of our choice. Let $\calC$ be any color class. Using sparsity, partition the set of links assigned to the squares in $\calC$ into at most $s$ subsets $S_1,S_2,\dots,S_k$, such the intersection of each $S_i$ and each square in $\calC$ is at most a single link. Then a standard area argument (see, e.g.~\cite{us:talg12}) shows that if the constant $c$ is sufficiently large, $S_i$ are feasible sets (e.g. under uniform power assignment). Note that it is important here that all links have length at most $\ell$, so they are ``attached'' to their corresponding squares.

Now consider a subset $S \subseteq L$ that is $s(L)$-dense, and let $\ell$ be the minimum link length in $S$,
and let $X$ be a $\ell$-by-$\ell$ square with $s(L)$ endpoints from $S$.
Let $T \subseteq S$ be the subset of links with endpoints in $X$, and note that $|T|\ge s(L)/2$.
The distance between any two points within $X$ is at most $\sqrt{2} \ell$. 
It follows that no pair of links in $T$ can coexist in a $\sqrt{2}^\alpha$-feasible slot.
That is, $T$, and therefore also $L$, requires $|T| \ge s(L)/2$ slots when $\beta \ge \sqrt{2}^\alpha$.
By signal strengthening, the exact value of $\beta$ changes the schedulability of the set only by a constant factor.
\end{proof}

\end{document}

%% file: MSTfig.tikz
\tikzstyle{vertex}=[circle,fill=black!25,minimum size=20pt,inner sep=0pt]
\tikzstyle{selected vertex} = [vertex, fill=red!24]
\tikzstyle{edge} = [draw,thick,-]
\tikzstyle{weight} = [font=\small]
\tikzstyle{sparse} = [draw,line width=2pt,-,orange!50]
\tikzstyle{ignored edge} = [draw,line width=5pt,-,black!20]

\begin{tikzpicture} [
    level 1/.style={sibling distance = 2cm, level distance = 2.5cm},
    level 2/.style={sibling distance = 1cm, level distance = 0.6cm},
    level 3/.style={sibling distance = 1cm, level distance = 0.6cm},
    every node/.style={circle, draw=black,thin, minimum size = 0.3cm},
    emph/.style={edge from parent/.style={red,very thick,draw}}, 
    empty/.style={edge from parent/.style={draw=none}},
    norm/.style={edge from parent/.style={blue,line width=2pt,draw}},
    mst/.style={edge from parent/.style={red,line width=2pt,draw}},	
  ]

  \begin{scope}[xshift=6cm]
    \node{} 
    child[norm] { node {}
      child[norm] { node {} 
		child[norm] { node {} 
			child[norm] { node {} 
				child {node[draw=none,fill=none] {$\vdots$}
					child[norm] { node (fe) {} }
					child[empty] { node[draw=none,fill=none] {} }	
				}
				child[empty] { node[draw=none,fill=none] {} }	
			}
			child[empty] { node[draw=none,fill=none] {} }	
	     }
	     child[empty] { node[draw=none,fill=none] {} }	
	}
	child[empty] { node[draw=none,fill=none] {} }	
    }
    child[mst] { node (ms) {}
      child[norm] { node {} 
		child[norm] { node {} 
			child[norm] { node {} 
				child {node[draw=none,fill=none] {$\vdots$}
					child[norm] { node (me) {} }
				}
			}
	     }
	}
    }
    child[mst] { node {}
	child[empty] { node[draw=none,fill=none] (ls) {} }
      child[norm] { node {} 
		 child[empty] { node[draw=none,fill=none] {} }
		child[norm] { node {} 
			child[empty] { node[draw=none,fill=none] {} }	
			child[norm] { node {}
				child[empty] { node[draw=none,fill=none] {} }	 
				child {node[draw=none,fill=none] {$\vdots$}
					child[empty] { node[draw=none,fill=none] {} }	
					child[norm] { node (le) {} }
				}
			}
	     }    	
	}	
    }
    ;
\path (ms) -- (ls) node[draw=none,fill=none] [midway] {$\ldots$};
\path[sparse,dotted] (me) -- (le);
\path[sparse] (me) -- (fe);
  \end{scope}
\end{tikzpicture}